\documentclass[aps,pra,
superscriptaddress,
floatfix,notitlepage,nofootinbib,longbibliography,11pt]{revtex4-1}
\usepackage{amsmath,amsthm,amsfonts,amssymb,braket}
\usepackage{graphicx}
\usepackage{color}

\makeatletter
\def\l@subsubsection#1#2{}
\makeatother

\usepackage[colorlinks=true,citecolor=blue,linkcolor=blue]{hyperref}
\usepackage[capitalize,nameinlink]{cleveref}

\newtheorem{lemma}{Lemma}[section]

\newtheorem{corollary}[lemma]{Corollary}

\newtheorem{claim}[lemma]{Claim}
\theoremstyle{definition}
\newtheorem{definition}[lemma]{Definition}
\newtheorem{remark}[lemma]{Remark}

\newcommand{\CC}{{\mathbb{C}}}
\newcommand{\RR}{{\mathbb{R}}}
\newcommand{\ZZ}{{\mathbb{Z}}}

\newcommand{\Udis}{U_{dis}}
\newcommand{\ambient}{M}
\newcommand{\cellfixed}{L}

\begin{document}
\title{Disentangling the Generalized Double Semion Model}

\author{Lukasz Fidkowski}
\affiliation{Department of Physics, University of Washington, Seattle WA 98195, USA}

\author{Jeongwan Haah}
\affiliation{Quantum Architectures and Computation, Microsoft Research, Redmond, WA 98052, USA}

\author{Matthew B.~Hastings}
\affiliation{Station Q, Microsoft Research, Santa Barbara, CA 93106-6105, USA}
\affiliation{Quantum Architectures and Computation, Microsoft Research, Redmond, WA 98052, USA}

\author{Nathanan Tantivasadakarn}
\affiliation{Department of Physics, Harvard University, Cambridge, MA 02138, USA}

\begin{abstract}
We analyze the class of Generalized Double Semion (GDS) models in arbitrary dimensions from the point of view of lattice Hamiltonians.
We show that on a $d$-dimensional spatial manifold $M$ the dual of the GDS is equivalent, up to constant depth local quantum circuits,
to a group cohomology theory tensored with lower dimensional cohomology models that depend on the manifold $M$.
We comment on the space-time topological quantum field theory (TQFT) interpretation of this result.
We also investigate the GDS in the presence of time reversal symmetry, 
showing that it forms a non-trivial symmetry enriched toric code phase in odd spatial dimensions.
\end{abstract}

\maketitle

\section{Introduction}

Gapped quantum phases of matter are, roughly speaking, equivalence classes of gapped lattice Hamiltonians 
under smooth deformation of parameters, in the thermodynamic limit of large system size.
Some (but not all~\cite{HaahFractons}) gapped quantum phases 
also possess effective topological quantum field theory (TQFT) descriptions~\cite{KHZ, LevinWen, Freed},
which capture their universal low energy features,
such as topological ground state degeneracy and quasiparticle braiding statistics.
TQFTs can be used to distinguish such gapped phases,
but it is not known whether they provide a complete set of invariants.

One class of TQFTs that arise from gapped lattice Hamiltonians and can be defined in arbitrary dimension 
are Dijkgraaf-Witten gauge theories~\cite{DijkgraafWitten}.
The generalized double semion model (GDS)~\cite{GDS} 
is another gapped lattice Hamiltonian that can be defined in arbitrary dimension.
Given the relative scarcity of TQFTs in spatial dimensions greater than $2$,
it is natural to ask what, if any, is the relation between the two.
It was shown~\cite{GDS} that for odd space dimension $d$ 
the GDS model is equivalent to a toric code up to a local quantum circuit, 
but in even dimensions $d \geq 4$ its ground state degeneracy does not match that of 
any $\ZZ_2$ Dijkgraaf-Witten theory~\cite{DijkgraafWitten}.
It was shown~\cite{Debray2018} that one could define a TQFT such that
on any closed manifold M, the state spaces of the TQFT match the space of low-energy states of the GDS model, 
and this isomorphism is compatible with the actions of the mapping class group of $M$ on both spaces.
This TQFT is a so-called ``gauge-gravity TQFT.'' 
This term is used because the action in the TQFT combines the $\ZZ_2$ gauge field 
with Stiefel-Whitney classes of the underlying manifold.
Roughly speaking, the action includes terms $\alpha^k w_{d+1-k}$ for each $k$, 
where $w_{d+1-k}$ is a Stiefel-Whitney class and 
$\alpha$ is the $\ZZ_2$ gauge field (see also~\cite{GEM18, STY18}).
In even spatial dimensions $d \geq 4$ this action differs from that of any $\ZZ_2$ Dijkgraaf-Witten theory,
so one might be tempted to conclude that in such dimensions the GDS represents a new gapped phase.

In this paper we analyze the GDS further from the point of view of a lattice Hamiltonian,
studying its equivalence under local constant depth quantum circuits to other lattice Hamiltonians.
Despite many challenges to making the definition of gapped phases of matter mathematically precise,
one case in which one can confidently say that two gapped Hamiltonians are in the same phase 
is if their ground states are related by a circuit of local unitary operators
that has constant depth, independent of system size.
Throughout, we consider Hamiltonians so that our ``spacetime'' is always some spatial manifold $M^d$ 
multiplied by a time direction.
Roughly, our results can be summarized as: 
the \emph{dual} of the GDS is equivalent, up to local quantum circuits, 
to several copies of group cohomology models (i.e., duals of Dijkgraaf-Witten models),
including one $d+1$-dimensional model, 
as well as possibly some additional lower dimensional models depending on the manifold~$M^d$.
In general, we have a group cohomology model in all spatial dimensions $i=0,1,\ldots,d$ 
for which $w_{d-i}$ is nontrivial, 
and the model can be chosen to be supported on a representative of the $i$-dimensional homology class of the manifold 
which is Poincar\'e dual to $w_{d-i}$.
These results then give some lattice interpretation of the TQFT action~\cite{Debray2018} 
as we can regard each term $\alpha^k w_{d+1-k}$ as arising from 
one of these lower dimensional group cohomology models, 
since a group cohomology model is dual to a Dijkgraaf-Witten model.

Here, the duality means the one introduced by Wegner~\cite{Wegner1971IsingGauge,Kogut1979LGT,LevinGu}.  Informally, it simply refers
to the inverse of the process of gauging a symmetry.  More precisely, rather than considering a theory of closed fluctuating $(d-1)$-cycles, 
we regard these $(d-1)$-cycles as bounding some spin configuration, 
and define a theory of these fluctuating spins.
Such a dual model has a global $\ZZ_2$ spin-flip symmetry, 
and the original model is recovered by ``gauging'' this $\ZZ_2$ symmetry~\cite{LevinGu}.
In general, any model with a global $\ZZ_2$ symmetry can be gauged,
and any model with a $\ZZ_2$ gauge charge --- 
i.e., an emergent bosonic quasiparticle with $\ZZ_2$ fusion rules --- 
can be ``un-gauged'' by condensing a bound state of this gauge charge with a local $\ZZ_2$ charge.
These processes are inverse to each other, 
so two models are in the same phase if and only if the same is true of their duals.\footnote{Technically, one also needs to keep track of which particle represents the gauge charge.}  The gauging and ungauging processes are also sometimes referred to as ``equivariantization" and ``de-equivariantization" respectively in the mathematical physics literature \cite{DGNO, ENO}.  The duals of Dijkgraaf-Witten and GDS models have constant depth disentangling circuits,
which makes them $\ZZ_2$ symmetry-protected topological (SPT) phases~\cite{LevinGu, Wen}.
We can thus frame our discussion entirely in the context of SPT phases, 
which are far better understood than general gapped phases~\cite{Wen, Kapustin_cobordism, FreedHopkins, CenkeXu}.  

In particular, it is predicted that 
there is a $\ZZ_2 \times \ZZ_2$ classification of $\ZZ_2$ SPT phases in $4+1$ spacetime dimensions, 
with the two $\ZZ_2$ generators referred to as the ``in-cohomology'' and ``beyond-cohomology'' phases
respectively~\cite{Wen2015, Kapustin_cobordism, Kapustin_pc, Freed_pc}.  This prediction comes from a computation of the
cobordism group $\Omega_5^{\rm{SO}} (B \ZZ_2)$ using a spectral sequence argument.  This group turns out to have order
$2^3 = 8$, and one can define $3$ corresponding $\ZZ_2$ invariants of orientable 5d manifolds by pairing the fundamental class $[M]$ with
$\alpha^5$, $w_2^2 \alpha$, and $w_2 w_3$ respectively~\cite{Debray_pc}.  The last invariant corresponds to an absolutely stable phase, whereas the first two
correspond to the in-cohomology and beyond-cohomology phases, respectively.  Therefore the in-cohomology phase is just the twisted Dijkgraaf-Witten dual,
while the effective action of the beyond cohomology phase 
is the same as the TQFT that describes the GDS~\cite{Debray2018}.
One might thus be tempted to conclude that the GDS dual is in the beyond-cohomology phase.
However, our results show that this is not the case: 
in flat space, where all the Stiefel-Whitney classes vanish,
the GDS dual is $\ZZ_2$ symmetric local circuit equivalent to the twisted Dijkgraaf-Witten dual instead.
On a general manifold, the GDS dual is equivalent 
to a twisted Dijkgraaf-Witten dual stacked with lower dimensional SPT phases 
in such a way that its spacetime response reproduces that of the beyond-cohomology phase.\footnote{We conjecture that the beyond cohomology phase is characterized by the following universal property: a pair of identical $\ZZ_2$ symmetry fluxes fuse to an odd number of $E_8$ states.  This is not true for the GDS dual, since it is certainly not true for a twisted Dijkgraaf Witten dual, and the two are equivalent in flat space.}

It should be understood, however, that this does not mean that the GDS
is a tensor product of several Dijkgraaf-Witten models,
as the dual of several cohomology phases is a single gauge theory.
In fact, the local reduced density matrices on any ball 
of the ground states of the GDS and Dijkgraaf-Witten theory
are mapped into each other by a local quantum circuit.
This follows because we can choose the lower dimensional subspace
where the cohomology states are located not to intersect any given ball
by symmetric local quantum circuits.

Our results may at first seem inconsistent 
with Ref.~\cite[Thm.8.1]{GDS} and Ref.~\cite[Thm.5.32]{Debray2018}, 
which state that the GDS TQFT in $4+1$ spacetime dimensions 
is inequivalent to any Dijkgraaf-Witten theory, 
whereas here we show that in flat space, 
the GDS Hamiltonian is circuit-equivalent to the twisted Dijkgraaf-Witten theory.
In fact, there is no inconsistency.
Indeed, as we just discussed, 
the GDS Hamiltonian is dual to the in-cohomology SPT phase 
stacked with lower dimensional SPT phases 
wrapping certain cycles of the spatial manifold.
It is these lower dimensional SPT phases 
that lead to the inequivalence discussed in Refs.~\onlinecite{GDS, Debray2018}.
In particular, when the spatial manifold is $\CC P^2$, 
there is a zero dimensional SPT --- i.e. a point charge ---
which eliminates all the ground states in the gauge theory, due to Gauss' law.

Many of our results have the flavor of proving two models equivalent by explicitly constructing a circuit.
We also attempt to show that some models are {\it non-equivalent} up to local circuits.
These kinds of questions have a long history.
For example, consider the two dimensional toric code.
It can be shown directly from the ground state degeneracy on a torus that the ground states of this model 
are not equivalent to product states up to a quantum circuit~\cite{BravyiHastingsVerstraete2006generation}.
However, to prove that the ground state is not equivalent to a product state on sphere 
is more difficult~\cite{Hastings2010Locality}; see \cite{Haah2014invariant} for a more general result.
Here, unfortunately, we will not be able to prove non-equivalence in many cases,
but we will be able to give some strong heuristic arguments for non-equivalence 
by relying on existing conjectured classifications of phases.

As an example of such a non-equivalence, consider again the case of odd space dimension.
While the GDS ground states in this case are equivalent to the toric code ground states 
up to a local quantum circuit, the circuit given in \cite{GDS} breaks time reversal symmetry 
in that it requires complex terms.  
We will argue, then, using existing results, that if one imposes an appropriate time reversal symmetry, 
the GDS in odd space dimensions is not equivalent to the toric code up to a local quantum circuit.
In the case of even space dimensions, 
we will show that one can remove the Dijkgraaf-Witten models supported on odd-dimensional homology 
if time reversal symmetry is absent, but we will argue that it is not possible otherwise, at least on ${\mathbb R}P^2$ (see Claim \ref{physargument1}).
Without time reversal symmetry, we will argue that it is impossible to remove the 2d in-cohomology SPT of unitary $\ZZ_2$ symmetry from
the non-trivial $S^2 \subset \CC P^2$ (see Claim \ref{physargument2}).

The remainder of this paper is structured as follows.  
In \cref{modelanddual}, we define the GDS model and its dual, 
constructing the unitary that disentangles the dual model, and we discuss the notion of 
{\it stable equivalence}, in which we tensor with additional degrees of freedom; stably equivalent models are physically equivalent.
In \cref{cohomologysection}, we define the group cohomology models.
While these models have been defined before, 
we consider a more general definition in which cohomology gates can act on an arbitrary closed $k$-chain, 
for some $k<d$, of some $d$-dimensional simplicial complex, 
and we show that homologically equivalent $k$-chains 
are related by local quantum circuits that commute with the group.
In \cref{sec:GDScircuitVScohomology}, we show that the unitary that disentangles the GDS dual model 
is given by some product of unitaries that disentangle cohomology models,
each acting on a closed $k$-chain. 
At this point, it remains only to determine what homology class these closed $k$-chains represent.
As a technical tool, it is very convenient to work on a simplicial complex 
that is a first barycentric subdivision; so, in \cref{sec:RG},
we show how one can use local quantum circuits to pass between different triangulations.
The notion of stable equivalence is useful here.

Finally, with all this background, the main result appears in \cref{combmanifoldsection}.
We pass to the first barycentric subdivision where we are able to identify that the cohomology gates 
act on chains dual to Stiefel-Whitney classes.
The main results are \cref{thm:GDSwTR} and \cref{thm:GDSwoTR}, 
corresponding to the cases with and without time reversal symmetry.
The notion of stable equivalence is useful here as it also allows one to,
informally, replace a product of cohomology gates on different closed $k$-chains by a tensor product of such gates, 
so that one may imagine that the dual model is equivalent to disconnected systems,
one for each $k$ with a nontrivial Stiefel-Whitney class. 
As an example, on $\CC P^2$, the GDS dual is equivalent to the tensor product of three models,
one of which is a four-dimensional cohomology phase on $\CC P^2$, 
one which is a two-dimensional cohomology phase on $\CC P^1 \subset \CC P^2$,
and one of which is a zero-dimensional cohomology phase, i.e., a single spin in the $|-\rangle$ state.

Regarding time reversal symmetry,
we will always work in the basis for local degrees of freedom such that the time reversal symmetry
is simply complex conjugation.

\section{GDS Model and Its Dual}
\label{modelanddual}
We review the GDS model and the duality transformation.
Perhaps the only subtle point arises when considering the relationship between stabilization 
(tensoring with additional degrees of freedom) and duality, 
especially when gauging a tensor product of SPT; see \cref{equivalencesection}.

\subsection{Review of GDS and Toric Code Models}
The toric code and GDS take as input a cellulation of a compact manifold $\ambient^d$; 
for the GDS, this cellulation is some
fixed Voronoi cellulation $\cellfixed$ generated by points in general position 
so that all cells in its dual cellulation are simplices.
The dual cellulation is called the Delaunay triangulation.  
In addition, for this paper we require that the Delaunay triangulation be a combinatorial manifold, 
i.e., a simplicial complex%
\footnote{
It will be useful for us to consider a simplicial complex purely combinatorially.
Given a set of vertices, a $k$-simplex is defined to be a set of $(k+1)$ vertices.
A face of a simplex is a nonempty proper subset of these vertices.
A simplicial complex is a collection of simplices whose faces are all contained in the collection.
}
where the link%
\footnote{
The link of a simplex $\sigma$ in a simplicial complex is a collection of simplices $\tau$
such that $\tau \cap \sigma = \emptyset$ but $\tau \cup \sigma$ is a simplex of the complex.
For a vertex $w_j$ to be in the link of $\Delta^{d-1}$ 
there must exists a $d$-simplex $\{w_j\} \cup \Delta^{d-1}$.
Heuristically, in a triangulation of a 3-manifold,
the link of a point is the 2-sphere that surrounds the point,
the link of a 1-simplex (line segment) is the circle that winds about the 1-simplex,
and the link of a 2-simplex (triangle) is the set of two points 
that are opposite to each other with the triangle in the middle.
}
of any simplex is homeomorphic to a sphere.%
\footnote{
It is an interesting question whether a Delaunay triangulation of a PL manifold $\ambient$ 
must always be a combinatorial manifold but for this paper we require that the triangulation has this property.
}

Each $(d-1)$-cell of of the cellulation has a single qubit, and we refer to a state $\ket \uparrow$ as ``absent" 
and $\ket \downarrow$ as ``present" 
(here our notation differs from \cite{GDS} since we use $\ket \downarrow$ to denote present).
The Hamiltonian for the GDS is
\begin{align}
H_{GDS}= H = H_+ + H_\square,
\end{align}
where
\[H_+ = \sum_{(d-2)\text{-cells }g} H_g, \qquad H_\square = \sum_\text{$d$-cells $c$} H_c.\]
Each term $H_g$ in $H_+$ is zero if an even number of $(d-1)$-cells 
meeting $g$ are $\downarrow$ (present)
and $1$ if an odd number are $\downarrow$.  
Then, the zero eigenspace of $H_+$ is spanned by closed $(d-1)$-chains with $\ZZ_2$ coefficients.
Each term $H_c$ in $H_\square$
is
defined to be
\begin{align}\label{Hcdef}
H_c=\frac{1-O_c}{2},
\end{align}
where
\begin{align}
\label{Ocdef}
O_c=\pm \prod_{\text{$(d-1)$ cells $f \in \partial c$}} X_{f}.
\end{align}
The operator $X_{f}$ is the Pauli-$X$ operator acting on the cell $f$.
The sign in \eqref{Ocdef} is $- (-1)^{\chi\left(\downarrow_c\right)}$ 
where $\chi$ is the Euler characteristic and $\downarrow_c$ is the codimension zero submanifold of $\partial c$ 
consisting of the union of ($d-1$) cells of $\partial c$ which are labeled $\downarrow$ in the state on which $H_c$ acts.  
The reason for choosing a generic cellulation is to make this subset a submanifold without self-intersections.
The terms of $H_+$ pairwise commute and commute with the terms in $H_\square$, while the terms in $H_\square$ pairwise commute in the eigenspace of $H_+$ with vanishing eigenvalue \cite{GDS}.

The toric code model with degrees of freedom on $(d-1)$-cell has the same Hamiltonian as the GDS, 
except that the sign in \cref{Ocdef} is always $-1$.  
Before considering the effect of this sign on the zero energy ground states\footnote{All Hamiltonians that we consider have non-negative spectrum (being sums of commuting terms with non-negative eigenvalues).  There is a possibility for a ground state to have non-zero energy (when the Hamiltonian is frustrated), but for the most part the ground states we consider will have zero energy.  We will hence use the terms `ground state' and `zero energy ground state' interchangeably.} of $H_{GDS}$, 
first recall the situation in the toric code.  
The zero energy ground states are superpositions of $(d-1)$-cycles.  More precisely,
the zero energy ground states are superpositions of ground states confined to fixed classes $x\in H_{d-1}(\ambient;\ZZ_2)$ 
and in each class the ground state is an equal amplitude superposition of all configurations.
For $d$ odd, the ground states of $H_{GDS}$ 
are also superpositions of ground states confined to fixed classes $x\in H_{d-1}(\ambient;\ZZ_2)$, 
but now the ground state is a superposition of cycles $E$ with amplitude $i^{\chi(E)}$.
For $d$ even, there may not exist zero energy ground states of $H_{GDS}$ in every homology sector.
However, if $\ambient$ is a $\ZZ_2$ homology sphere, 
an explicit formula for the zero energy ground state amplitude on cycle $E$ can be given as $(-1)^{s(E)}$ 
where $s(E)$ is a {\it semi-characteristic}~\cite{GDS}.
Since the formula for the ground state amplitudes 
is fairly complicated for the GDS in even dimensions,
this motivates considering the dual model, 
which simplifies much of the treatment and for which the ground state amplitudes have a simple expression on all manifolds.

\subsection{Duality}

Informally, duality means that, rather than considering a theory of closed fluctuating $(d-1)$-cycles, 
we regard these $(d-1)$-cycles as bounding some spin configuration, 
i.e., in the trivial homology sector of the GDS, we write the cycle as a boundary of some $d$-chain.  
Then, we define a theory of these fluctuating spins. 
Formally, duality is defined as an isomorphism between two algebras; physically it is the inverse process to gauging.
Duality can be defined for an arbitrary group $G$,
but for this subsection we consider a tailored construction for $G=\ZZ_2$.

Consider a system of qubits identified with $d$-cells of some generic finite cell complex,
referred to as the \emph{primal} complex in this subsection.%
\footnote{
so that the Poincar\'e dual (Delaunay) is a simplicial complex.
}
We require that the topological space underlying the complex be connected.
Let algebra~$\mathcal A$ be generated by $Z_c Z_d$ and by $X_c$ where $c,d$ label $d$-cells, 
i.e., the algebra generated by even products of Pauli $Z$ operators and by arbitrary products of Pauli $X$ operators.
Now consider a different system of qubits identified with $(d-1)$-cells of the complex.
Let $\mathcal B^+$ be the algebra of operators acting on the qubits on the $(d-1)$-cells,  
generated by operators $Z_f$ for any $(d-1)$-cell $f$ 
and $\prod_{  \text{$(d-1)$ cells }f \in \partial c}  X_f$ for any $d$-cell $c$.
Let $\mathcal B = \mathcal B^+/ \mathcal J$ be the quotient algebra of $\mathcal B^+$ 
by the relation that
given any $1$-cycle on the dual complex, the product of $Z_{e'}$ over that cycle 
($e'$ is a $(d-1)$-cell in the primal cellulation) be equal to $+1$.
If the dual complex has trivial first homology, 
the two-sided ideal $\mathcal J$ of relations is generated by 
$-I + \prod_{(d-1)\text{-cells }f~:~ g \in \partial f} Z_f$ for arbitrary $(d-2)$-cell $g$
of the primal complex.

We define the duality map between $\mathcal A$ and $\mathcal B$ by
\begin{align}
\label{dualityeq}
Z_c Z_d  & \longleftrightarrow  Z_{f} &\text{if } (\partial c) \cap (\partial d) = f,  \\ 
X_c &\longleftrightarrow \prod_{(d-1)\text{-cells }f \in \partial c} X_f &\text{for any $d$-cell } c. \nonumber
\end{align}
This defines a map $\mathcal A \to \mathcal B$ uniquely: 
Since the space is connected,
given a product $Z_c Z_d$ in ${\mathcal A}$ for arbitrary $c,d$,
there is a path $c,c',c'',\ldots,d$ from $c$ to $d$ consisting of $d$-cells
on the $1$-skeleton of the dual complex.
Then, $Z_c Z_{c'}$ is mapped following \cref{dualityeq} 
to some operator in $\mathcal B^+$, as is $Z_{c'} Z_{c''}$, and so on.
Any two paths from $c$ and $d$ results in the same operator modulo $\mathcal J$
since $\mathcal J$ is precisely generated 
by products of $Z$ along the difference of such paths.
The inverse map $\mathcal B^+ / \mathcal J \to \mathcal A$
is well-defined for a similar reason.

We say that the algebra ${\mathcal A}$ is the algebra for the dual model 
and the algebra ${\mathcal B}$ is the algebra for the gauge model.
The map from ${\mathcal A}$ to ${\mathcal B}$ is referred to as ``gauging.''

\subsection{Disentangling Circuit}
We now apply this duality to the GDS.  When applying this duality, it is convenient to take the Poincar\'e dual of the complex; then, the qubits in $H_{GDS}$ are identified with $1$-cells, i.e., edges, and the qubits in the dual theory will be identified with vertices.

The operator $H_+$ modulo $\mathcal J$ is dual to a scalar, equal to zero.
The terms $O_c$ in $H_\square$ are dual to $\pm X_c$; 
this is clear since the second line of \cref{dualityeq} 
gives that $X_c$ is dual to $\prod_{(d-1)\text{-cells }f \in \partial c} X_f$ 
and the sign in $O_c$ is mapped to some diagonal operator, i.e., some other sign.

We claim (and show in \cref{Ocdual}) that the sign is such that the dual of
$O_c$ is equal to
$\Udis X_c \Udis^\dagger,$
where $\Udis$ is a diagonal unitary with eigenvalues $\pm 1$ (hence, $\Udis=\Udis^\dagger$) given by
\begin{align}
\Udis=(-1)^{\chi(\downarrow_\cellfixed)},
\end{align}
where $\downarrow_\cellfixed$ is a codimension zero subcomplex of the Voronoi cellulation $\cellfixed$ consisting of all $d$ cells labelled $\downarrow$.
We refer to $\Udis$ as the {\bf GDS dual disentangler}.
Then the ground state of the GDS dual Hamiltonian is the image under the GDS dual disentangler of the state with all qubits in the $|+\rangle$ state.

We now describe the unitary $\Udis$ in terms of local (but not $\ZZ_2$ symmetric) gates.
The spins of the GDS dual are on the vertices of the Delaunay triangulation.
Define a sequence of operators $Z,CZ,CCZ,\ldots$.
The operator $Z$ refers to the Pauli $Z$ operator on a qubit.
The operator $CZ$ is a controlled-$Z$ operator defined to be 
the diagonal operator on the two qubits which is $-1$ if both qubits are labelled $\downarrow$ and which is $+1$ otherwise.
Generally $C^kZ$ is a diagonal operator acting on $k+1$ qubits 
which is $-1$ if all qubits are labelled $\downarrow$ and $+1$ otherwise.
By the additivity formula for the Euler characteristic, 
$\Udis$ is equal to the product, 
over all $k$-simplices for $0\leq k\leq d$ of the Delaunay triangulation, 
of $C^k Z$ on that simplex.

We now show that (see also \cite{LevinGu})
\begin{lemma}\label{Ocdual}
The dual of $O_c$ is equal to $\Udis X_c \Udis$.
\begin{proof}
We have 
$O_c=-(-1)^{\chi\left(\downarrow_c\right)} \prod_{\text{$(d-1)$ cells $f \in \partial c$}} X_{f}$.
The dual of $\prod_{\text{$(d-1)$ cells $f \in \partial c$}} X_{f}$ is $X_c$.
We compute $\Udis X_c \Udis$ acting on some configuration of down spins $C$; 
we write the corresponding state $|C\rangle$.
We have $\Udis X_c \Udis=\pm |C'\rangle$, where $C'$ is obtained from $C$ by flipping the spin at cell $c$.  
Assume without loss of generality that the spin at cell $c$ is down in $C$.
We have $\chi(C')-\chi(C)=\chi(c)-\chi(c \cap C)$.
We have $\chi(c)=1$ since $c$ is a $d$-ball.
$c \cap C$ is a union of $(d-1)$-cells $f\in \partial c$ such that each $f$ is attached to two $d$-cells: 
one $d$-cell is $c$ and the other $d$-cell is some cell in $C$ that we denote $n(f)$.
Informally, the duality says that each such $f$ is a boundary between a down and up spin in the GDS dual,
and so it represents a down spin (``present") in the GDS.
Formally, the operator $Z_c Z_{n(f)}$ is dual to $Z_f$.
Hence, $(-1)^{\chi(c \cap C)}$ is dual to 
$(-1)^{\chi\left(\downarrow_c\right)}$.
Since $(-1)^{\chi(c)}=-1$, we have
$\Udis X_c \Udis$ dual to $O_c$.
\end{proof}
\end{lemma}

Also we consider the commutation of $X$ with $\Udis$, showing that they commute up to a sign determined by the Euler characteristic of the Voronoi cellulation $\cellfixed$:
\begin{lemma}\label{lem:chiflip}
Let $X$ be the global spin flip.  Then,
\begin{align}
X \Udis X \Udis = (-1)^{\chi(\cellfixed)}.
\end{align}
\begin{proof}
Let us first give a proof using properties of a combinatorial manifold before giving an alternative elementary proof.
We have $X \Udis X \Udis=(-1)^{\chi(\uparrow_\cellfixed)} (-1)^{\chi(\downarrow_\cellfixed)}$, 
where $\uparrow_\cellfixed$ consists of all $d$-cells labelled $\uparrow$.
By the additivity formula for the Euler characteristic,
$(-1)^{\chi(\uparrow_\cellfixed)} (-1)^{\chi(\downarrow_\cellfixed)}=(-1)^{\chi(\cellfixed)} (-1)^{\chi(\partial \downarrow_\cellfixed)}$.
If $\downarrow_\cellfixed$ is a manifold, then
since $\partial \downarrow_\cellfixed$ is a boundary, it has even Euler characteristic.

We now give the elementary proof.  
For any $C^k Z$ acting on any $k$-simplex $\Delta$, 
the conjugation $X (C^k Z) X$ is equal to $-1$ times 
the product of $C^j Z$ over all simplices of dimension $j=0,1,\ldots,k$ in $\Delta$.
Since $X \left[ \prod (C^k Z) \right] X = \prod \left[ X (C^k Z)X \right]$,
this in particular implies that 
$X \left[ \prod (C^k Z) \right] X \left[ \prod (C^k Z) \right] 
= \prod \left[ X (C^k Z)X (C^k Z) \right]$
as diagonal operators commute.
Then, we see in the group commutator $X \Udis X \Udis$,
taking the product of $C^k Z$ over all simplices in $\cellfixed$, 
on any given $j$-simplex the terms $C^j Z$ will cancel
if and only if that $j$-simplex is a face of an even number of higher dimensional simplices.
The number of such higher dimensional simplices is equal modulo $2$ 
to the Euler characteristic of the link, and so vanishes modulo $2$ 
for a combinatorial manifold.
The factors of $-1$ for each $\Delta$ give a factor of $(-1)^{\chi(\cellfixed)}.$
\end{proof}
\end{lemma}
Hence, if $\chi(\cellfixed)$ is even (odd), the ground state of the GDS dual is even (odd) under $X$.
It follows then that $O_c$ commutes with $X$; of course, this was to be expected from the properties of the duality: every operator in ${\cal A}$ commutes with $X$.

\subsection{$G$-Equivalence, Stabilization, and Duality}
\label{equivalencesection}
In this paper, we will consider {\it stable} equivalence.  Let us define this first for a dual model, i.e. a model invariant under some group $G$.  First we need to define the notion of equivalence under local quantum circuits.  Here we will generalize to an arbitrary group $G$; in this case we will define a ``computational basis", where the basis states of each qudit are labelled by group elements and the action of the group in this basis is by group multiplication.

A {\bf $G$-invariant circuit} is a local quantum circuit whose gates are invariant under $G$.
One can define a local quantum circuit formally if one wishes by considering families of unitaries and requiring that the unitaries in the circuit have depth and range which are both $O(1)$.
The dual of a $G$-invariant circuit is a circuit which leaves $H_+$ invariant.  We say that
two unitaries $U,V$ are  {\bf $G$-circuit equivalent} if the unitary $U^\dagger V$ can be realized by a $G$-invariant circuit.

We say that
two unitaries $U,V$ are {\bf stably $G$-equivalent} if the unitaries $(U\otimes I)$ and  $(V \otimes I)$ are $G$-circuit equivalent, where $U\otimes I$ denotes $U$ tensored with the identity matrix on some number of additional qudits, and $V \otimes I$ denotes $V$ tensored with the identity matrix on some possibly different number of additional qudits.

When tensoring with additional qudits for stable equivalence, we will consider some refinement $\cellfixed'$ of the original cellulation $\cellfixed$.   Each $d$-cell $\Delta$ of $\cellfixed$ corresponds to some qudit in the first tensor factor (i.e., in the factor acted on by $U$ or $V$).  There are one or more $d$-cells in $\cellfixed'$ contained in $\Delta$; the added $d$-cells correspond to the added qudits and we (arbitrarily) choose one of these $d$-cells to correspond to a qudit in the first tensor factor.

We can also define stable $G$-equivalence of states, saying that
two pure states, $\psi_1,\psi_2$, are {\bf $G$-equivalent} if $\psi_1=U \psi_2$
for some unitary $U$ which is a $G$-invariant circuit.
Two pure states, $\psi_1,\psi_2$, are  stably $G$-equivalent if $\psi_1 \otimes |+\rangle^{n_1}$ is $G$-equivalent to
 $\psi_2 \otimes |+\rangle^{n_2}$,
where $|+\rangle^{n_1}$ and $|+\rangle^{n_2}$ denote 
some number of additional qudits in a product state which is invariant under the symmetry group~$G$.
For $G=\ZZ_2$, the state $|+\rangle$ is simply the usual $+1$ eigenstate of Pauli $X$.

Note that while a toric code is dual to a model with $\ZZ_2$ symmetry, two copies of the toric code  is dual to a model with the symmetry $\ZZ_2 \times \ZZ_2$ since one has a symmetry in each copy.  However, one copy of the toric code on a bilayer is dual to a bilayer with $\ZZ_2$ symmetry.

If the state $\psi$ is a zero energy ground state of some $G$-invariant Hamiltonian $H_{dual}$ which is a sum of commuting projectors, then $\psi\otimes |+\rangle$ is also the ground state of some $G$-invariant Hamiltonian $H^+_{dual}$, as one can add a $G$-invariant term to force the added qudit into the $|+\rangle$ state.  There is an obvious generalization to the case that one adds several qudits.
A more natural way of defining a $G$-invariant Hamiltonian with added qudits is to ``copy" the states of the qudit when refining the cell structure, so that all qudits corresponding to cells in the refinement ``have the same state as the qudit in the original complex".  Formally, for each $d$-cell $\Delta$ in $\cellfixed$, define an isometry $V_\Delta$ as follows: let $\Delta'_1,\ldots,\Delta'_j$ for some integer $j$ be the $d$-cells contained in $\Delta$ in some refinement $\cellfixed'$
and let $V_\Delta$ be the isometry
$V_\Delta=\sum_{g\in G} |g,\ldots,g\rangle \langle g|,$ where the bra is the state of the qudit on $\Delta$ and the ket is the states of the qudits corresponding to $\Delta'_1,\ldots,\Delta'_j$.  Let $V=\prod_\Delta V_\Delta$ so that the domain of $V$ has one qudit per cell in the original complex and the codomain has one qudit per cell in the refinement.
Let $H'_{dual}=V H_{dual} V^\dagger + \sum_{\Delta} (1-V^\dagger_\Delta V_\Delta).$  The second term adds a penalty unless all qudits in a refinement of a given cell $\Delta$ are in the same state.  We leave it to the reader to show that $H^+_{dual}$ and $H'_{dual}$ are related by local quantum circuits.

Suppose $H_{dual}$ is dual to some Hamiltonian $H_{gauge}$ and $H'_{dual}$ is dual to $H'_{gauge}$.  
Here, the added qudit will represent some added $d$-cell in a refinement of the original cellulation, and so $H'_{gauge}$ may have some number of added $(d-1)$-cells.
Then, given any zero energy ground state of $H_{dual}$, one can construct a zero energy ground state of $H'_{gauge}$ by tensoring in additional degrees of freedom on the $(d-1)$-cells and then acting with a local quantum circuit.  In this case, all $(d-1)$-cells bounding a pair of cells $\Delta'_1,\Delta'_2\in \cellfixed'$ which are both contained in the same cell $\Delta\in\cellfixed$ will be in the empty state; this is enforced by the term dual to $\sum_{\Delta} (1-V^\dagger_\Delta V_\Delta)$.
For any pair of different cells $\Delta_1,\Delta_2 \in \cellfixed$, all $(d-1)$-cells in the refinement which are in the boundary of both $\Delta_1,\Delta_2$ will be in the same state as each other.

\section{Cohomology Models}
\label{cohomologysection}
Here we review a construction for symmetry-protected topological (SPT) states based on group cohomology~\cite{DijkgraafWitten,Wen}.
Although the underlying math is virtually identical,
our exhibition here differs from existing ones
in that we consider cohomology states on subsystems defined by any homological cycles
whereas in usual treatment, Ref.~\cite{Wen} in particular,
authors have considered SPT states on the (top dimensional) fundamental homology class only.
This extension enables us to address cohomology states on physical spaces of more general topology.
We complement the construction of cohomology states 
with an (inverse) entanglement renormalization group transformation in \cref{sec:RG}.

\begin{definition}[Group cochain circuit]\label{def:cochaincircuit}
Given a group cochain $\omega: G^{k+2} \to \RR / \ZZ$
we define a diagonal unitary (quantum gate) $U_\omega : (\CC G)^{\otimes (k+1)} \to (\CC G)^{\otimes (k+1)}$,
called a {\bf cochain gate}, as
\begin{align}
U_\omega \ket{g_0,g_1,\ldots,g_k} = \exp\left(2 \pi i \omega(e,g_0,g_1,\ldots,g_k)\right) \ket{g_0,g_1,\ldots,g_k} 
\end{align}
where $e \in G$ is the group identity element.
For an integral simplicial $k$-chain $C = \sum_j a_j \Delta_j$ with $a_j \in \ZZ$
in a simplicial complex where each $k$-simplex $\Delta_j$ has a fixed ordering of vertices
by which it is oriented,
the {\bf cochain circuit} by $\omega$ on $C$ 
is the product of all the cochain gates $U_{a_j \omega, \Delta_j}$ over $k$-simplices $\Delta_j$,
where the ordering of vertices is used to match the qudits in each simplex 
with the coordinates of the argument of $\omega$.
Since the gates commute with each other the product is unambiguous.
\end{definition}

Recall that the coboundary $\delta$ on the group cochain complex is defined as
\begin{align}
(\delta \omega)(g_0,\ldots,g_{k+2}) = \sum_{j=0}^{k+2} (-1)^j \omega(g_0, \ldots, \hat g_j , \ldots, g_{k+2})
\end{align}
where $\hat g_j$ means to omit $g_j$.
Let $\bar g$ for $g \in G$ be a symmetry operator on the full physical system.
If $G$ is an internal on-site symmetry, the action of $G$ on $\RR/\ZZ$ is trivial and
$\bar g$ is the tensor product of unitaries $\ket h \mapsto \ket{gh}$ over all degrees of freedom.
If $G$ is an antiunitary symmetry, $G$ acts nontrivially on $\RR/\ZZ$ 
and $\bar g$ acts on the phase $\eta \in \RR/\ZZ$ as $e^{2\pi i \eta} \mapsto e^{2\pi i g \eta}$.%
\footnote{
Strictly following this construction, 
a basis state $\ket g$ of a local degree of freedom must be mapped to an orthogonal basis state under any nonidentity symmetry action.
This means that the time reversal symmetry
should be represented as the global spin flip followed by complex conjugation, 
which corresponds to the diagonal subgroup of our $\ZZ_2 \times \ZZ_2^T$ later.
However, our time reversal symmetry is just the complex conjugation.
This should not cause any confusion as our explicit states are always $\ZZ_2 \times \ZZ_2^T$ symmetric.
}
A $(k+1)$-cochain $\omega$ is {\bf homogeneous} 
if $g \omega(\vec x) = \omega(g \vec x)$ for any $\vec x \in G^{k+2}$ and $g \in G$.

\begin{lemma}[Cycle-commutativity]\label{lem:cyclecomm}
Let $\omega:G^{k+2} \to \RR/\ZZ$ be a nonzero homogeneous cocycle.
Let $U$ be a cochain circuit by $\omega$ on a simplicial $k$-chain $C$.
Then, the commutator $\bar g U \bar g^{-1} U^\dagger$ is equal to a cochain circuit 
by $\omega(g,\cdot)$ on $\partial C$.
In particular, for a cocycle $\omega$ valued in the cyclic subgroup $\ZZ/n$ of $\RR/\ZZ$,
the commutator vanishes if and only if $\partial C = 0 \bmod n$.
Here, unlike usual simplicial homology, 
the boundary of a $0$-chain $\sum_j a_j \Delta^0_j$ is $\sum_j a_j$ (remark: this is sometimes called reduced homology).
\end{lemma}
\begin{proof}
Since $U_{\omega,\Delta}$ and $\bar g U_{\omega,\Delta} \bar g^{-1}$ commute because they are diagonal unitaries,
the overall commutator is the product of all (local) commutators between cochain gates and $\bar g$.
Hence, to prove the first claim it suffices by linearity 
to verify it for a single cochain gate $U_{\omega,\Delta}$.
The local commutator is a diagonal gate $\exp[2\pi i (g \omega(e,g^{-1}\vec x) - \omega(e,\vec x))]$.
Here the phase is equal to $\eta = \omega(g,\vec x) - \omega(e,\vec x)$ because $\omega$ is homogeneous.
The cocycle condition
$(\delta \omega)(g,e,\vec x) = 0$ implies that the phase is equal to
\begin{align}
\eta = \omega(g,\vec x) - \omega(e,\vec x) = \sum_{j=0}^k (-1)^j \omega(g,e,x_0,\ldots, \hat x_j, \ldots,x_k).
\end{align}
But the simplex $\Delta = (v_0,\ldots,v_k)$ spanned by vertices $v_0,\ldots,v_k$ has the boundary chain
$\partial \Delta = \sum_j (-1)^j (v_0,\ldots, \hat v_j, \ldots,v_k)$.
Hence, the phase $\eta$ is equal to the sum of phases from the cochain circuit by $\omega(g,\cdot)$ 
on the simplicial $(k-1)$-chain $\partial \Delta$. 

The second claim follows obviously from the first.
\end{proof}

\begin{lemma}[Coboundary circuit is locally symmetric]\label{lem:coboudarylocalsymmetric}
A $G$-cochain circuit by a homogeneous coboundary $\omega = \delta \lambda$ on a simplicial cycle $C$ 
is equal to a product of gates, each of which is $G$-symmetric.
\end{lemma}
\begin{proof}
The coboundary operation gives 
\begin{align}
\omega_k(e,\vec x) = \lambda(\vec x)- \underbrace{\sum_{j=0}^{k} (-1)^j \lambda(e,x_0,\ldots,\hat x_j,\ldots,x_k)}_{\lambda'(\vec x)}.
\end{align}
The gate given by the first term $\exp(2\pi i \lambda(\vec x))$ is manifestly $G$-symmetric since $\lambda$ is homogeneous.
The gate given by the second term $\lambda'$ on a simplex $\Delta$ is precisely the cochain gate by $\lambda'$ on $\partial \Delta$.
Hence, if $C = \sum_a c_a \Delta_a$ is the simplicial cycle,
the product $\prod_{a} \left(e^{2\pi i c_a\lambda'(\vec x)}\right)_{\Delta_a}$ is the identity because $C$ is closed.
\end{proof}

\begin{lemma}[Cocycle circuits on homologous chains are equivalent]\label{lem:moveWithinHomologyclass}
A $G$-cochain circuit by a $k$-cocycle $\omega$ on a $k$-chain $\partial C$ that is a boundary 
is equal to a product of gates, each of which is $G$-symmetric and is supported on a $(k+1)$-simplex.
\end{lemma}
\begin{proof}
The simplicial chain $C = \sum_a c_a \Delta_a$ is a collection of $(k+1)$-simplices.
The circuit is a product of ``local'' circuits by $\omega$ on $c_a \partial \Delta_a$,
each of which is $G$-symmetric by \cref{lem:cyclecomm} since $\partial \Delta_a$ is closed.
\end{proof}

\begin{lemma}[Cocycle circuits generate invertible states]\label{lem:invertible}
Let $V$ be a $G$-cochain circuit by a cocycle $\omega$ on a simplicial $k$-chain $C$.
Suppose $\omega$ is valued in $\ZZ/n \subset \RR/\ZZ$ and $\partial C = 0 \bmod n$.
Then, the tensor product $V \otimes V^\dagger$ is a product of gates, 
each of which is on $2k+2$ qudits and is $G$-symmetric.
\end{lemma}

The tensor product $V \otimes V^\dagger$ is a unitary 
which acts on a tensor product of two copies of the system, 
so that it acts on a system of qudits on the Voronoi cellulation 
$\cellfixed \sqcup \cellfixed$.  
Further, $V\otimes V^\dagger$ is a $G$-cochain circuit by cocycle $\omega$ 
on a chain equal to the sum of $C$ on the first copy of $\cellfixed$ 
and $-C$ on the second copy.

\begin{proof}
Let $\Delta = \{ v_0,\ldots, v_k \}$ be any simplex of 
$C=\sum_a c_a \Delta_a$ and $\Delta' = \{ v'_0,\ldots, v'_k \}$ its copy.
We construct a $(k+1)$-simplicial chain $P \Delta$ (called the prism operator)
on vertices $\Delta \sqcup \Delta'$ by 
$P \Delta = \sum_{j=0}^k (-1)^j \{ v'_0, \ldots, v'_j, v_j,\ldots,v_k \}$.
It is routine to check that $\Delta - \Delta' = \partial P \Delta + P \partial \Delta$.
Define $P(C) = \sum_a c_a P(\Delta_a)$.
Then, by linearity $\partial P(C) + P(\partial C) = C - C'$ as integral chains.
Since $\partial C = 0 \bmod n$, the $\omega$-circuit on $\partial P(C)$ is $V \otimes V^\dagger$,
which has locally symmetric decomposition by \cref{lem:moveWithinHomologyclass}.
\end{proof}

\section{Equivalence of GDS Dual with Product of Cohomology Phases}
\label{sec:GDScircuitVScohomology}

We now tailor our discussion to the symmetry group that contains an internal symmetry group~$\ZZ_2$, the overall flip.
Let us fix a nontrivial homogeneous representative $\omega_k$ of the unique nonzero cohomology class in $H^k(\ZZ_2; \ZZ_2 \subset \RR/\ZZ) = \ZZ_2$ as%
\footnote{
The cochain can shown to be a cocycle by direct computation;
in order for the coboundary of $\omega$ to assume a nonzero value
there should not be any triple repetition in the argument,
but then any double repetition in the argument yields two terms that cancel with each other, and
alternating arguments gives two nonzero terms which cancel.
The chain $\omega$ is not a coboundary since $\sum_{\vec x} (\delta \lambda)(e,\vec x) = 0 \bmod 1$ 
for any homogeneous $\lambda:G^{k+1} \to \{0, 1/2\}$.
}
\begin{align}
\begin{cases}
\omega_k(X,e,X,e,\ldots) = \frac 1 2\\
\omega_k(e,X,e,X,\ldots) = \frac 1 2\\ 
\omega_k(\text{any other}) = 0
\end{cases},
\label{eq:omegak}
\end{align}
where $\ZZ_2 = \{e, X \}$.  We emphasize that we use this representative for both even and odd $k$.
Note that the cochain gate $U_{\omega_k,\Delta^k}$ has one and only one eigenvalue of $-1$.
We wish to consider all possible compositions of cochain gates by $\omega_k$.
For a system of $d+1$ qubits, we define $\mathcal U_d$ to be 
the group of all diagonal unitaries (in the standard basis of $\CC[\ZZ_2]^{\otimes(d+1)}$)
whose eigenvalues are all $\pm 1$.
\begin{lemma}\label{lem:omegaGenerateUk}
$\mathcal U_d$ is generated by cochain gates associated with $\omega_k$ for $k = 0,1,\ldots,d$
on various sets of qubits, and $\pm I$.
\end{lemma}
\begin{proof}
Let $\{ \ket 0 , \ket 1 \}$ be a basis of $\CC[\ZZ_2]$.
Consider a unitary in $\mathcal U_d$ which has all diagonal entries equal to $1$ except for one single diagonal entry equal to $-1$ on $\ket{y_0,y_1,\ldots,y_d}$.
Then this unitary has eigenvalue $e^{i \pi (x_0-y_0-1) \cdots (x_d-y_d-1)}$ 
on $\ket{x_0,\ldots,x_d}$.
Since these unitaries generate all elements of $\mathcal U_d$, 
we see that any element of $\mathcal U_d$ gives a multivariate polynomial
with binary coefficients modulo relations $x_j^2 = x_j$, 
and in turn, any such polynomial gives an element of $\mathcal U_d$.
So, there is a one-to-one correspondence (actually a group isomorphism)
between $\mathcal U_d$ and a quotient polynomial ring 
$R = \ZZ_2[x_0,\ldots,x_d]/(x_0^2-x_0,\ldots,x_d^2-x_d)$
with the multiplication in the polynomial ring forgotten.

Each cochain gate is then identified with
$x_0 (1+x_1)x_2 \cdots x_k$ for $k$ even or
$x_0 (1+x_1)x_2 \cdots (x_k + 1)$ for $k$ odd,
a multivariate polynomial of degree%
\footnote{
Strictly speaking, we do not have a ring homomorphism from $R$ to $\ZZ_2[x]$
that will allow us to speak of the degree in the usual sense.
Nonetheless, here we use the term ``degree'' of a monomial 
to mean the number of distinct variables in it,
and the degree of an element of $R$ is the maximum monomial degree over its all terms.
} $k+1$.
So, we must show that the span over $\ZZ_2$ of polynomials identified with cochain gates
and scalar $1$ that comes from $\pm I$ in the assumption,
includes all of $R$.
We proceed by induction in the degree.
Degree zero polynomials are clearly in the span of $1$.
Degree one monomials appear in the span of 
(the polynomials of the cochain gates associated with) $\omega_0$,
and thus all degree one polynomials are in the span of $1, \omega_0$.
If all polynomials of degree $k-1$ are in the span of $1,\omega_0,\ldots,\omega_{k-2}$,
then, since polynomials of $\omega_{k-1}$ contain all possible monomials of degree $k$,
all polynomials of degree $k$ are in the span of $1,\omega_0,\ldots,\omega_{k-1}$.
\end{proof}

\begin{lemma}[Cohomology gates generate all globally $\ZZ_2$-symmetric diagonal circuits]
\label{lem:allcohomology}
Let $U$ be a quantum circuit of diagonal gates, 
each of which acts on at most $d+1$ qubits $\CC [\ZZ_2] = \CC^2$ and has eigenvalues $\pm 1$.
Suppose $U$ commutes with $\bar X$ (the global flip).
Then, $U$ is equal to a product $V_0 \cdots V_d$ of cochain circuits,
where the cochain circuit $V_k$ is by $\omega_k$ on a $\ZZ_2$-closed simplicial $k$-chain
for $k = 0,\ldots, d$. 
Here, the homological boundary of a point is nonzero but the empty set (reduced homology).
The depth of $V_0 \cdots V_d$ (the number of layers of nonoverlapping gates) 
is at most the depth of $U$ times a finite function of $d$.
\end{lemma}
This lemma takes place in an abstract simplicial complex.  Note that by \cref{lem:cyclecomm} each of $V_0,\ldots,V_d$ commutes with $\bar X$.
\begin{proof}
By \cref{lem:omegaGenerateUk} each gate of the circuit $U$ can be written 
as a product of cochain gates associated with $\omega_k$ for various $k$.
Here, a simplex on which a cochain gate acts is defined simply as a collection of qubits.
Thus, we may assume that $U$ is a circuit consisting of cochain gates.
This may blow up the depth of the circuit, but not by more than a finite factor that depends only on $d$.

Consider the $\ZZ_2$-simplicial $d$-chain $C_d$ given by the sum of all simplices 
on which $\omega_d$-gates are supported.
In \cref{lem:cyclecomm} 
we have seen that taking the commutator with the symmetry operator pushes the dimension of the gate down by one (the dimension of the gate is the dimension of the simplex on which the gate acts).
Hence, in the commutator $U \bar X U \bar X$, the $(d-1)$-dimensional part (and by induction every dimensional part) must vanish,
which implies $\partial C_k = 0 \bmod 2$ for all $k = d, d-1,\ldots,0$.
\end{proof}

Note that even if the system of qubits is embedded in a manifold, 
the number $d$ in the lemma above may be larger than the manifold's topological dimension;
e.g., on the two-dimensional square lattice a gate may act on four qubits that comprise a square.
However, any such simplicial chain, if closed, is necessarily null-homologous.  
Note also that in \cref{lem:omegaGenerateUk}, 
if the polynomial identified with $\mathcal U_d$ is a sum of monomials on simplices of a complex,
then $\mathcal U_d$ is generated by cochain gates on simplices of that complex.

\section{Entanglement RG}\label{sec:RG}

The (inverse) entanglement RG transformation
is simply a quantum circuit whose gates are individually symmetric.
Some new degrees of freedom can be added, but they are initialized in the manifestly symmetric state.
Since the GDS dual or any cohomology state is defined on a simplicial complex 
rather than on a more general cell complex,
we are going to define a sequence of ``moves'' from a triangulation to another triangulation.
It will be useful for us to focus on moves 
such that the composition of all the moves is the transition from the original triangulation
to its barycentric subdivision.

\begin{figure}
\centering
\includegraphics[width=0.9\textwidth, trim = {0ex 90ex 50ex 0ex}, clip]{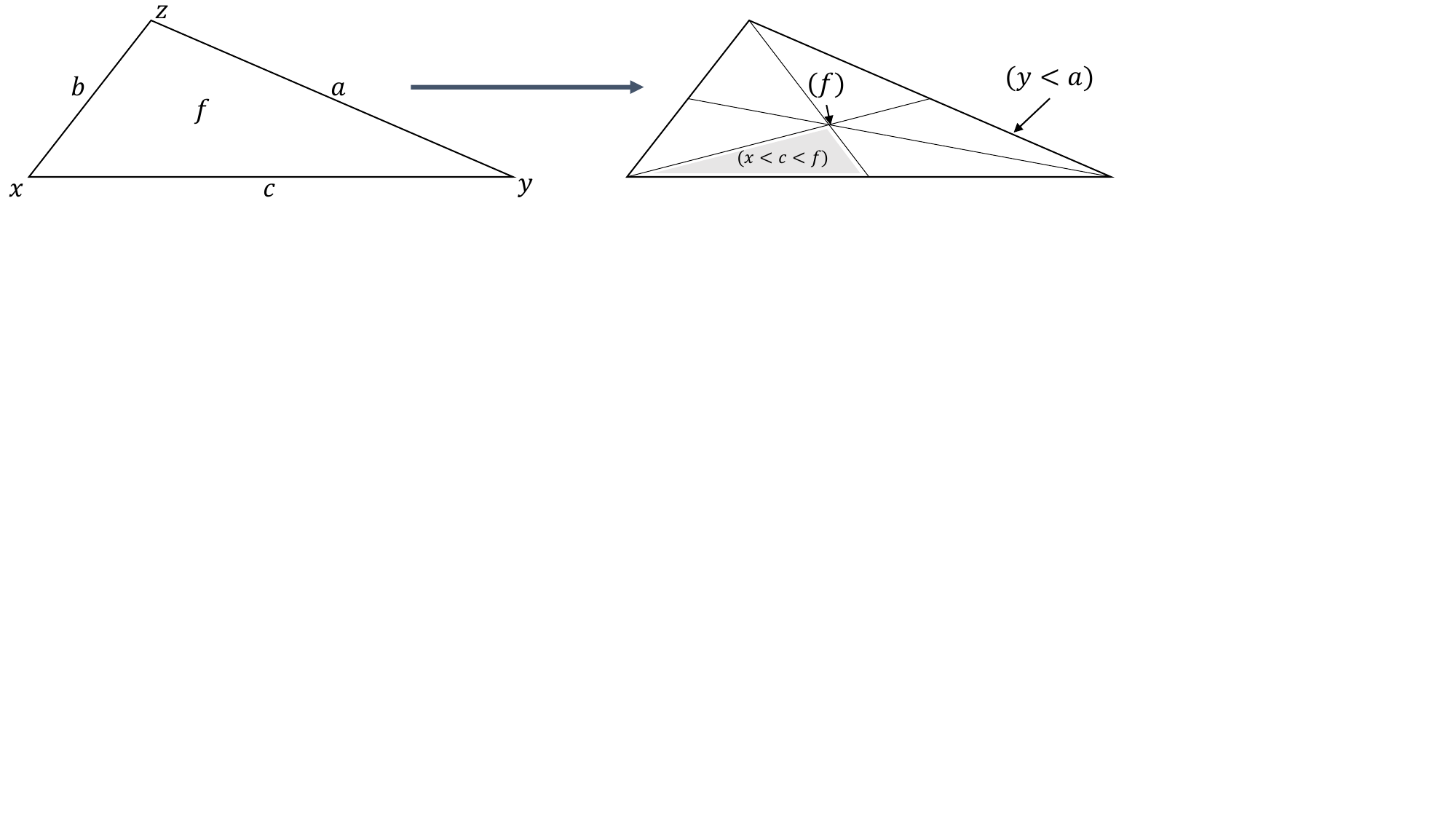}
\caption{Barycentric subdivision of a $2$-simplex.
$x,y,z$ are vertices, $a,b,c$ are edges, and $f$ is a face.
}
\label{fig:bcsd}
\end{figure}

\begin{figure}
\centering
\includegraphics[width=0.9\textwidth, trim = {0ex 80ex 30ex 0ex}, clip]{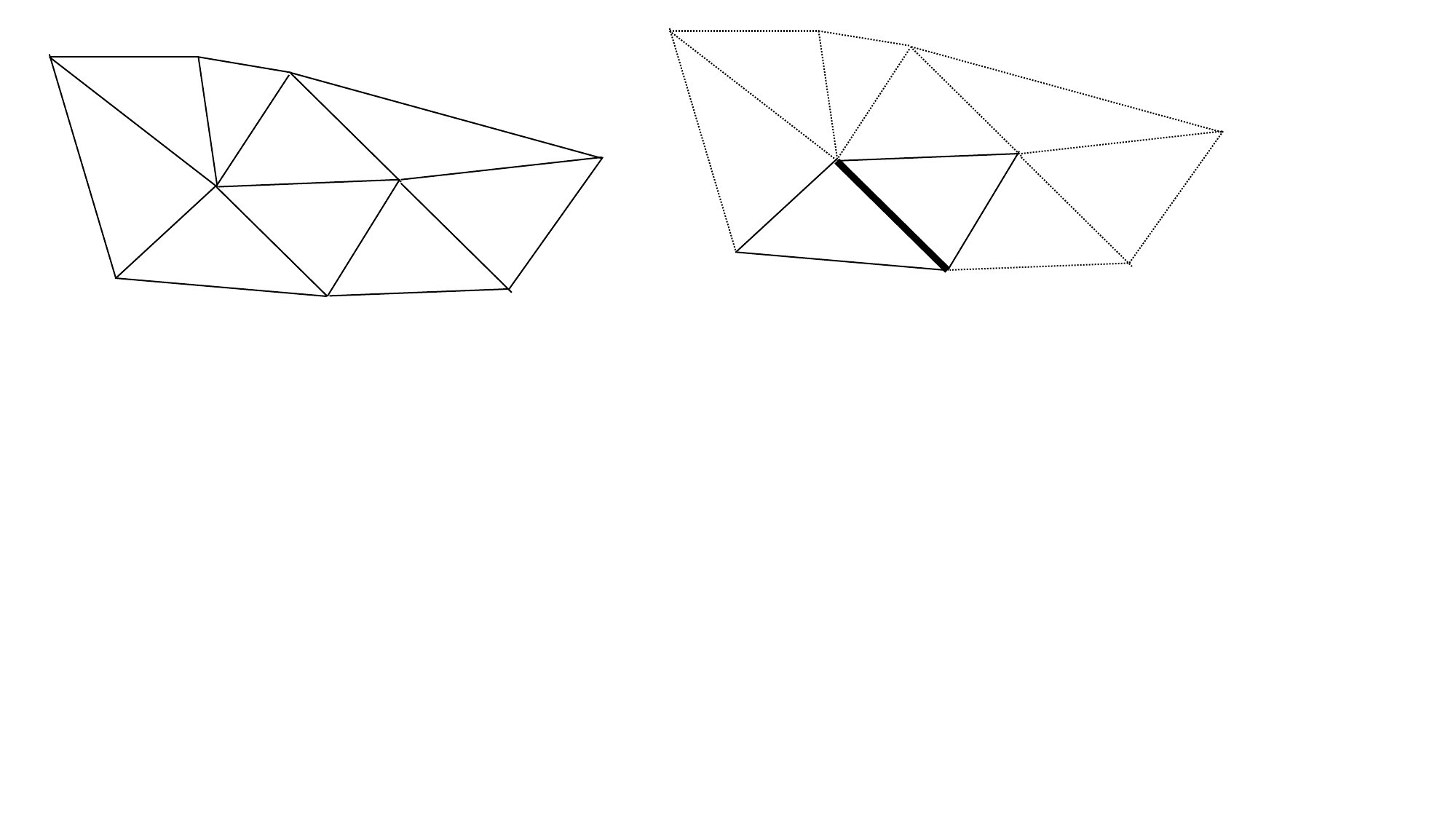}
\caption{A $2$-dimensional simplicial complex and the closed star of a $1$-simplex.}
\label{fig:closedstar}
\end{figure}

\begin{figure}
\centering
\includegraphics[width=0.9\textwidth, trim = {0ex 110ex 45ex 0ex}, clip]{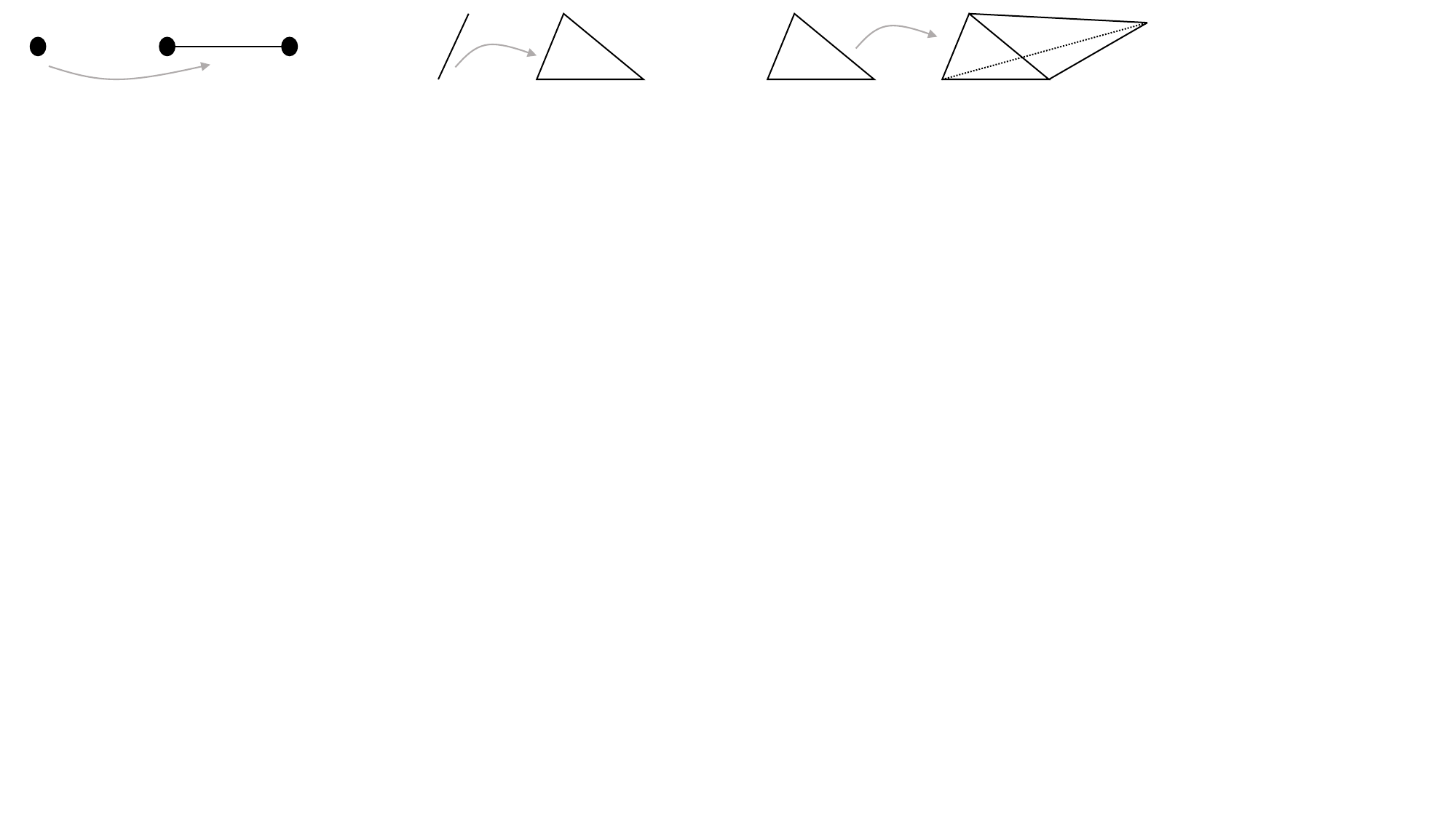}
\caption{The cone over a $k$-simplex is a $(k+1)$-simplex.}
\label{fig:cones}
\end{figure}

Before we present our construction it may be useful to remind standard notions in topology.
The {\bf barycentric subdivision} of a simplicial complex $L$ is a simplicial complex $L'$
where the $k$-simplicies of $L'$ are in one-to-one correspondence with the sequences 
$(\sigma_0 < \sigma_1 < \cdots < \sigma_k)$ of simplices of $L$.
Here $\sigma < \tau$ means that $\sigma$ is a proper nonempty face of $\tau$.
(See \cref{fig:bcsd} for the barycentric subdivision of a $2$-simplex.)
In particular, a vertex of $L'$ is identified with (the barycenter of) a simplex of $L$.
For example, in a triangulation of a surface, the barycentric subdivision has three ``types'' of vertices 
corresponding to the dimension of the simplex $\sigma_0$ of $L$,
and three ``types'' of edges:
$(\Delta^0 < \Delta^1)$ that connects a vertex of $L$ to the barycenter of an edge of $L$,
$(\Delta^0 < \Delta^2)$ that connects a vertex of $L$ to the barycenter of a triangle of $L$,
and
$(\Delta^1 < \Delta^2)$ that connects the barycenter of an edge of $L$ to the barycenter of a triangle of $L$.
(See the right figure of \cref{fig:bcsd}.)
The {\bf star} of a simplex $\Delta$ is the union of all simplices that have $\Delta$ as a face.
The {\bf closed star} of a simplex is the smallest simplicial complex that includes the star.
(See \cref{fig:closedstar}.)
A {\bf cone} over a space $\mathcal Y$ is the quotient space of $\mathcal Y \times [0,1]$
under the identification of $\mathcal Y \times \{ 1 \}$ as a single point (apex).
A cone over a $k$-simplex is a $(k+1)$-simplex.
(See \cref{fig:cones}.)
For a subset $\mathcal X$ of a space $\mathcal Y$,
a point $y \in \mathcal Y$ is in the {\bf topological boundary} of $\mathcal X$ in $\mathcal Y$ 
if and only if
every open neighborhood of $y$ intersects both $\mathcal X$ and $\mathcal Y \setminus \mathcal X$.

\subsection{Barycentric subdivision by local moves}

We design two kinds of ``local'' moves from any simplicial complex $L$ to another simplicial complex
with the same underlying space such that the composition of these moves turns $L$ into its barycentric subdivision
$L'$. Let $d$ be the dimension of $L$.
It may be instructive to see \cref{fig:t1t2} as we go along.

\begin{figure}
\centering
\includegraphics[width=\textwidth, trim = {0ex 70ex 0ex 0ex}, clip]{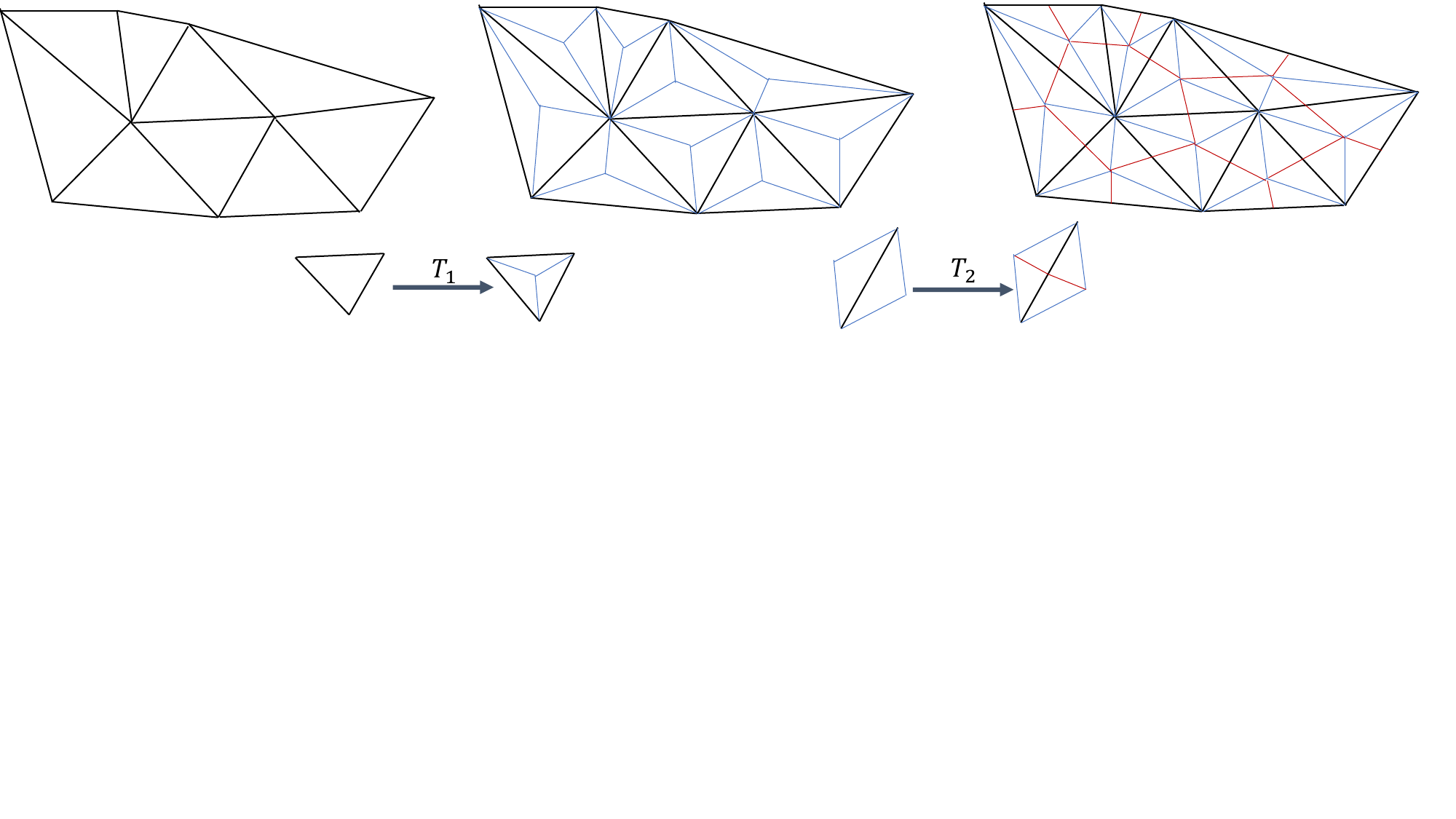}
\caption{Barycentric subdivision of a triangulation by local moves.}
\label{fig:t1t2}
\end{figure}

The first move is to replace a $d$-simplex
$\Delta^d = \{ v_0,\ldots, v_d \}$
with a $d$-dimensional simplicial complex $T_1(\Delta^d)$ with $d+2$ vertices $\{ v_0,\ldots, v_d, (\Delta^d) \}$.
The additional vertex $(\Delta^d)$ should be depicted in mind 
as one that sits at the barycenter of $\Delta^d$;
it will become the vertex $(\Delta^d)$ of the barycentric subdivision
and thus we used the same notation.
Every proper face of $\Delta^d$ is again a face of $T_1(\Delta^d)$,
and every $d$-dimensional simplex of $T_1(\Delta^d)$ 
is $\{ v_{j_0},\ldots, v_{j_{d-1}}, (\Delta^d) \}$
which is a cone over a $(d-1)$-face of $\Delta^d$
where the apex is the additional vertex $(\Delta^d)$.
Thus, $T_1(\Delta^d)$ has $(d+1)$ simplices of dimension $d$.
Note that the boundaries of $\Delta^d$ and $T_1(\Delta^d)$ are the identical simplicial complexes.

The second move is to replace the closed star
of a $(d-1)$-simplex $\Delta^{d-1}$ of $L$ with a simplicial complex $T_2(\Delta^{d-1})$.
The closed star of $\Delta^{d-1}$ is a union of cones
over $\Delta^{d-1}$, one for each vertex $t_j$ in the link
of $\Delta^{d-1}$ in $L$.
The complex $T_2(\Delta^{d-1})$ is defined to be the triangulation of this union
derived from the barycentric subdivision of $\Delta^{d-1}$ in the obvious way;
it includes all the simplices of the barycentric subdivision of $\Delta^{d-1}$
as well as all their cones, the apex of each is a vertex $t_j$ in the link of $\Delta^{d-1}$ in $L$.
As in the first move, the new complex $T_2(\Delta^{d-1})$ and the closed star of $\Delta^{d-1}$ 
have the same simplicial complex on their topological boundary
in $|L| = |L'|$;
no subdivision was made on the topological boundary.

The promised sequence of moves is to apply $T_1$ for every $d$-simplex of $L$ to obtain $L^\circ$,
which includes the $(d-1)$-skeleton of $L$ (and more),
and then $T_2$ for every $(d-1)$-simplex of $L$ where the star is taken in $L^\circ$.
The second kind of move should not be applied to any additional $(d-1)$-simplex that is introduced by $T_1$,
and each vertex $t_j$ of the previous paragraph is a barycenter of some $d$-simplex of $L$.
It is clear from the construction that the resulting simplicial complex is
the barycentric subdivision $L'$ of the original simplicial complex $L$.

\subsection{Mapping to barycentric subdivision}

Having constructed the sequence of moves between simplicial complexes,
we consider the difference of the GDS dual or cohomology states along the moves,
measured in the generating circuits.
For each move $T_1$ or $T_2$, the change occurs in a local simplicial subcomplex $A$
that involves one or more $d$-simplices.
The difference of the circuits is to remove all the gates on $A$
and introduce gates on $T(A)$ where $T = T_1$ or $T_2$.
Let $L$ be a $d$-dimensional simplicial complex such that every vertex is contained in at most $R$ simplices.
Let $L'$ be the barycentric subdivision of $L$.

\begin{lemma}\label{lem:RG-cohomology}
Let $C$ be a $k$-chain of $L$ such that $\partial C = 0 \bmod n$,
and $\omega$ be a $k$-cocycle of a group $G$ valued in $\ZZ/n \subset \RR/\ZZ$.
Let $\ket \psi \otimes \overline{\ket +}$ denote the cochain state $\ket \psi$ by $\omega$ on $C$
tensored with $\ket + = \sum_{g \in G}\ket g$ over all barycenters of nonzero dimensional simplices of $L$.
Then, there is a locally $G$-symmetric quantum circuit on $L'$ of depth $\le$ a finite function of $R$,
under which $\ket \psi \otimes \overline{\ket +}$
is mapped to the cochain state by $\omega$ on a $k$-chain $C'$ of $L'$ homologous to $C$.
\end{lemma}
\begin{proof}
Let $A$ be the subcomplex of $L$ where a local move by $T = T_1$ or $T = T_2$ is happening:
the subcomplex $A$ consists of a single $d$-simplex for $T_1$ 
(and of course all its faces)
or of a union of all cones over a $(d-1)$-simplex for $T_2$.
Let $L^A$ be the simplicial complex obtained from $L$ by applying the local move on $A$.

Put $C = \sum_{\Delta \in L} a_\Delta \Delta$ where $a_\Delta \in \ZZ$,
and $C|_A = \sum_{\Delta \in A} a_\Delta \Delta$.
Each $\Delta^k \in L$ is canonically mapped into a chain $D(\Delta^k)$ of $L^A$ 
as the sum of all $k$-simplices of $L^A$ that are contained in $\Delta$
where the orientation of simplices of $L^A$ is such that
the homological boundary map $\partial$ and $D$ commute.
Define the chain $C^A$ of $L^A$ as the canonical image of $C$.

Let us remark that a simplicial chain 
$\sum_\Delta a_\Delta \Delta$
for a cochain circuit specifies 
ordered tuples of vertices $\{\Delta_j\}$ of $\Delta$ and coefficients $a_\Delta \in \ZZ/n$.
See \cref{def:cochaincircuit}.
Hence, we may speak of a cochain circuit 
on \emph{a formal linear combination of ordered tuples of vertices},
regardless of whether the ordered tuples of vertices are simplices of
an ambient simplicial complex.
We take this view for the rest of this proof.

Then, the difference $V(A)$ in the generating circuits, 
one by $\omega$ on $C$ and the other by $\omega$ on $C^A$,
is the product of a cochain circuit by $\omega$ on $-C|_A$ and that on $C^A|_{T(A)}$; 
in other words, $V(A)$ is a cochain circuit by $\omega$ on $C^A|_{T(A)} - C|_A$.

Regardless of whether we consider $T_1$ or $T_2$,
the homological boundary $\partial (C|_A) \bmod n$ is supported on 
the topological boundary of $|A|$ in $|L| = |L^A|$;
otherwise, the $\bmod n$ homological boundary of $C$ would not be empty.
But in either move, $\partial (C|_A) = \partial (C^A|_{T(A)}) \bmod n$ 
since no subdivision is made on the topological boundary of $|A|$.
Hence, $V(A)$ is on a chain that is homologically closed modulo $n$ 
(in an abstract simplicial complex%
\footnote{
This is actually a $k$-sphere. 
This fact is used, rather obviously, in the proof of \cref{lem:RG-GDS}.
} 
defined by all simplices that appear in
the formal linear combination $C^A|_{T(A)} - C|_A$ 
of ordered tuples of vertices),
and \cref{lem:cyclecomm} implies that $V(A)$ is $G$-symmetric.

The transition from $C$ to $C'$ is induced 
by rounds of moves
on nonintersecting subcomplexes, where the number of rounds depends on $R$ only.
This means that the promised quantum circuit consists of the local gates of form $V(A)$.
\end{proof}

\begin{lemma}\label{lem:RG-GDS}
Let $G = \ZZ_2 \times \ZZ_2^T$.
Suppose that the link of any $(d-1)$-simplex consists of two vertices.
Then, there is a locally $G$-symmetric quantum circuit on $L'$ of depth $\le$ a finite function of $R$,
under which the GDS dual state on $L$ tensored with $\ket +$ on all the barycenters of nonzero dimensional simplices 
is mapped to the GDS dual state on $L'$.
\end{lemma}
\begin{proof}
The proof is parallel to that of \cref{lem:RG-cohomology}.
Using the same notation there, 
let us first consider a move $T_1$ on a subcomplex $A_1$ 
that has only one $d$-simplex $\Delta^d$.
Then, the difference $V(A_1)$ in the dual GDS circuits
is the product of (i) $C^{\le d}Z$ on all simplices of $A_1$
and (ii) $C^{\le d} Z$ on all simplices of $T_1(\Delta^d)$.
Since no subdivision is made on the boundary of $\Delta^d$ by $T_1$,
the gates of form $C^{<d} Z$ on the boundary of $\Delta^d$ cancel.
Thus, $V(A_1)$ is equal to the product of (1) the dual GDS circuit on 
a $d$-sphere whose northern hemisphere is triangulated by $A_1$ 
and the southern hemisphere by $T_1(A_1)$, 
and (2) the dual GDS circuit on the equator of the $d$-sphere.
Since the boundary of $A_1$ matches that of $T_1(A_1)$ as noted earlier,
the two triangulations of the northern and southern hemispheres agree along the equator.
Now, \cref{lem:chiflip} says that the GDS dual circuit is $X$- and TR-symmetric on any sphere,
and $V(A_1)$ is symmetric as desired.

For a move $T_2$, we have a subcomplex $A_2$ that includes all cones 
over a given $(d-1)$-simplex $\Delta^{d-1}$ of $L$.
The assumption that the link of every $(d-1)$-simplex within $L$ consists of two vertices,
implies that after all $T_1$ moves for all $d$-simplices of $L$,
there are exactly two $d$-simplices in $A_2$.
Hence, the space $A_2$ glued along its boundary with $T_2(A_2)$ is a $d$-sphere.
The factorization of $V(A_2)$ as two dual GDS circuits, 
one on the $d$-sphere and the other on its equator, is again valid.
Therefore, $V(A_2)$ is symmetric by \cref{lem:chiflip}.

The promised circuit consists of gates of form $V(A)$ for some subcomplex $A$.
\end{proof}

\section{GDS on a Combinatorial Manifold}
\label{combmanifoldsection}
A combinatorial manifold is a simplicial complex where the link of any simplex is homeomorphic to a sphere,
which implies that the underlying topological space is indeed a manifold.
As we have seen earlier, the dual GDS state on this simplicial complex is created from the all $\ket +$ state
by $Z,CZ,CCZ,\ldots, C^dZ$ gates on every simplex; 
$Z$ on every qubit, $CZ$ on every edge, $CCZ$ on every triangle, etc.
In particular, the circuit that disentangles the dual GDS state
satisfies the assumptions of \cref{lem:allcohomology}.

Therefore, the state can be thought of as a product of cohomology states.
It remains to be determined where the cohomology states live, 
especially those that are created by lower dimensional gates than the spatial dimension.
By \cref{lem:moveWithinHomologyclass}, only the homology class of the support of the cohomology state
is important.
We will determine the homology class using the stable equivalence of SPT states;
we will refine the lattice by adding more degrees of freedom and consider the states there.

\subsection{Passing to the barycentric subdivision}

The underlying topological space of a simplicial complex 
is unchanged by taking barycentric subdivision,
and hence on general grounds no properties of a topological model should change;
that is, taking barycentric subdivision should be unessential.
Indeed, we have shown that the GDS dual state on a combinatorial manifold
can be mapped to that on the barycentric subdivision by a circuit that is locally symmetric 
under both the global $\ZZ_2$ spin flip and time reversal.
However, it makes the technical manipulation far easier
to take the first barycentric subdivision.

The technical convenience of passing to the barycentric subdivision
comes from the following facts that are stated in lemmas below.
Let us partially order the vertices $(\Delta)$ of the barycentric subdivision $L'$ of $L$
by the face-relation: $(\sigma) \le (\tau)$ iff $\sigma$ is a face of $\tau$.
This partial order gives a total order among vertices of a simplex of $L'$,
and hence defines the orientation of each simplex.
Recall that a simplex of $L'$ is labeled by a sequence 
$(\sigma_0 < \ldots < \sigma_k)$
of simplices of $L$.
We define an integral simplicial chain $C_k$ of the barycentric subdivision $L'$ for each $k = 0,1,\ldots,d$ by
\begin{align}
C_k = \sum_{\sigma_0 < \cdots < \sigma_k} (-1)^{|\sigma_0| + \cdots + |\sigma_k|} (\sigma_0 < \cdots < \sigma_k)
\end{align}
where the sum ranges over all $k$-simplices of $L'$.
\begin{lemma}[\cite{HalperinToledo}]\label{lem:HalperinToledo}
Suppose that $L$ is a combinatorial $d$-manifold. Then, for all $k = 1,\ldots,d$
\begin{align}
\partial C_k = \chi(S^{d-k}) C_{k-1}
\end{align}
where $\chi(S^k)$ is the Euler characteristic of the $k$-sphere.
Moreover, the chain $C_k \bmod 2$ represents the $(\bmod 2)$ Stiefel-Whitney homology class for $k = 0,\ldots, d$.\footnote{Recall that the $k$th Stiefel-Whitney homology class is the Poincare dual of $w_{d-k}$, the $(d-k)$th Stiefel-Whitney cohomology class.}
Here the $d$-th Stiefel-Whitney homology class is the $\ZZ_2$-fundamental homology class.
\end{lemma}

\begin{lemma}\label{lem:cochainEqualsGDS}
Suppose the link of any simplex in $L$ has even Euler characteristic.
For any $k = 0,1,\ldots,d$,
the cochain circuit by $\omega_k$ of \cref{eq:omegak} on $C_k$,
is equal to the product of $C^k Z$ gates over all $k$-simplices of $L'$.
\end{lemma}
In Ref.~\cite{Sullivan} the local even Euler characteristic condition was used to show
that $C_k$ is closed modulo $2$.
\begin{proof}
The cochain gate by $\omega_k$ on a simplex $\Delta^k = (\sigma_0< \cdots <\sigma_k) \in L'$
multiplies a phase factor $(-1)^p$ with $p = x_0(x_1+1)x_2 \cdots (x_k \text{ or } x_k + 1)$
where $x_j = 0,1$ are binary variables representing the state of the qubits in $\Delta^k$.
The last factor in $p$ is $x_k$ if $k$ is even or $x_k +1$ if $k$ is odd.
A gate $C^k Z$ on a $k$-simplex multiplies a phase factor $(-1)^{x_0 x_1 \cdots x_k}$.
Hence, for $k = 0$ the claim is obvious. From now on, $k > 0$.

The phase difference $p - (x_0 x_1 \cdots x_k)$ is a sum of monomials,
each of which represents some $C^{< k}Z$ gate on specific faces of $\Delta^k$.
These faces have labels that are obtained by omitting one or more $\sigma_j$ with odd $j$ 
from $(\sigma_0 < \cdots < \sigma_k)$.
Thus, the number of $C^{k'}Z$ gates where $\lfloor k/2 \rfloor \le k' < k$ 
that are applied to a $k'$-simplex $\Delta^{k'} = (\tau_0 < \cdots < \tau_{k'})$
is the number $N(\Delta^{k'})$ of $k$-simplices of form $(\sigma_0 < \cdots < \sigma_k)$
where the sequence $(\tau_0 < \cdots < \tau_{k'})$ 
is obtained from $(\sigma_0< \cdots <\sigma_k)$
by omitting one or more $\sigma_j$ with odd $j$.
We claim that $N(\Delta^{k'})$ is even for all $\Delta^{k'}$,
which will conclude the proof.

We count $N(\Delta^{k'})$ by partitioning the collection of all the relevant $k$-simplices
by fixing all omitted $\sigma_j$ but the first $\sigma_{j_0}$.
For example, if $k'=3$, $k=5$, and
$\Delta^{3} = (\sigma_0 < \sigma_2 < \sigma_3 < \sigma_4)$
in a 6-dimensional simplicial complex,
then we partition $5$-simplices $(\sigma_0 < \sigma_1 < \sigma_2 < \sigma_3 < \sigma_4 < \sigma_5)$
according to $\sigma_5$ so that in each subcollection the $5$-simplices differ in $\sigma_{j_0} = \sigma_1$ only.
Then, each subcollection of $k$-simplices is identified with the link of 
$\sigma_{j_0 - 1}$ within the simplicial complex $\partial \sigma_{j_0 + 1}$;
if $j_0 = k$ then the link is taken in $L$.

Now, if $j_0 < k$, then the link of $\sigma_{j_0 - 1}$ within the sphere $\partial \sigma_{j_0 + 1}$
is always a lower dimensional sphere, which has an even number of simplices.
If $j_0 = k$, then
the even Euler characteristic assumption implies that 
the link of $\sigma_{k-1}$ has an even number of simplices.
\end{proof}

\subsection{With Time Reversal Symmetry}
With time reversal symmetry, our main result is the following:
\begin{corollary}\label{thm:GDSwTR}
Let $G = \ZZ_2 \times \ZZ_2^T$ be the global symmetry group.
The GDS dual disentangler on a closed combinatorial $d$-manifold $L$
is stably $G$-equivalent to the product $V_0 \cdots V_d$
of cochain circuits $V_k$ by $\omega_k$ on the Stiefel-Whitney homology class
for $k=0,1,\ldots,d$.

The product $V_0 V_1\cdots V_d$ is stably $G$-equivalent to the tensor product 
$V_0 \otimes V_1 \otimes \cdots \otimes V_d$ acting on $d+1$ copies of $L$.
\end{corollary}
\begin{proof}
By \cref{lem:RG-GDS,lem:RG-cohomology} we may assume without loss of generality
that $L$ is a barycentric subdivision of some other combinatorial manifold.
By \cref{lem:cochainEqualsGDS} we can regard the circuit as a product of cochain circuits by $\omega_k$
on $C_k$, but \cref{lem:HalperinToledo} says $C_k$ are representatives of Stiefel-Whitney classes.
The very representative is unimportant because of \cref{lem:moveWithinHomologyclass}.

The second claim follows from \cref{lem:invertible} since $C_k$ are closed modulo $2$.
\end{proof}

\begin{remark}
This lemma gives a way to identify Stiefel-Whitney clasess of a
of a combinatorial manifold without passing to the barycentric subdivision.  \cref{lem:cochainEqualsGDS} gives some product of cochain circuits by $\omega_k$
on $C_k$ and then the $C_k$ will be representantives of the Stiefel-Whitney classes.  We expect that this will reproduce the formula of \cite{goldstein1976formula}.
\end{remark}

Since the circuit $V_d$ for odd $d$ is believed to create a symmetry protected phase of time reversal and $\ZZ_2$ symmetry, this suggests that the GDS in odd spatial dimensions is
not equivalent to the toric code up to a local time reversal invariant circuit, i.e., that it is a symmetry enriched topological phase.
We now consider whether it is possible to remove $V_k$ for odd $k<d$.

Let us consider the example $L={\mathbb{R}}P^2$.
The GDS dual ground state consists of a stack of $0$-, $1$-, and $2$-dimensional cohomology states.
The $0$-dimensional state is just a $\ZZ_2$ charge, 
and certainly cannot be removed with a symmetric local circuit.
However, an interesting question is whether the $1$-dimensional state, 
which is the Haldane phase of $\ZZ_2 \times \ZZ_2^T$ symmetry,
can be removed with a constant depth local circuit of $\ZZ_2 \times \ZZ_2^T$-symmetric local unitaries.
We argue that it cannot:

\begin{claim}\label{physargument1}
Under the physical assumption that 
there are no nontrivial invertible phases for bosons in $1d$ 
and that all $1d$ $\ZZ_2 \times \ZZ_2^T$ SPT phases have $\ZZ_2$ fusion rules 
(i.e. two copies of any such phase is equivalent to the trivial phase),
the Haldane phase of $\ZZ_2 \times \ZZ_2^T$ symmetry 
cannot be removed from the non-trivial cycle of ${\mathbb{R}}P^2$ 
with a constant depth local circuit of $\ZZ_2 \times \ZZ_2^T$-symmetric local unitaries.
\end{claim}

We proceed as follows: think of ${\mathbb{R}}P^2$ as a M\"obius band glued to a disc along their common boundary.
The Haldane phase runs along the $S^1$ at the center of the M\"obius band.
Any constant depth circuit can be written as a circuit on the disc composed with a circuit on the M\"obius band,
up to a thickening of these regions of order the range (Lieb-Robinson length) of the circuit.
A circuit acting on the disc can only produce an invertible $\ZZ_2 \times \ZZ_2^T$-symmetric phase
along the disc's boundary.
It is believed that there are no nontrivial invertible phases for bosons in 
$1d$~\cite{Kitaev_pc, Freed, FreedHopkins, GaiottoFreyd},
so this must actually be an SPT phase of $\ZZ_2 \times \ZZ_2^T$.
But all such SPT phases have $\ZZ_2$ fusion rules~\cite{Wen,Kapustin_cobordism},
so when one collapses the M\"obius band onto $S^1$,
the $S^1$ is in a trivial SPT phase.
No local symmetric circuit acting on just the M\"obius band can change this fact,
so the Haldane phase on the center $S^1$ cannot be removed, as desired.
Below we will more carefully explain a higher dimensional version of this argument 
for the case without time reversal symmetry.

\subsection{Without Time Reversal Symmetry}
Without time reversal symmetry, our main result is the following:
\begin{corollary}\label{thm:GDSwoTR}
Let $G = \ZZ_2$ be the global symmetry group.
The GDS dual disentangler on a combinatorial $d$-manifold $L$ 
is stably $G$-equivalent to the identity if $d$ is odd, 
or the product $V_0 V_2 \cdots V_d$
of cochain circuits $V_k$ by $\omega_k$ on the Stiefel-Whitney homology classes of $L$ if $d$ is even.

When $d$ is even, the product $V_0 V_2 \cdots V_d$ is stably $G$-equivalent to the tensor product 
$V_0 \otimes V_2 \otimes \cdots V_d$ acting on $\frac d 2 + 1$ copies of $L$.
\end{corollary}
\begin{proof}
We start with \cref{thm:GDSwTR}. 
Note that $\omega_{2k+1} = \frac 1 2 \delta \omega_{2k}$ which can be checked by direct computation.

If $d$ is even, we have to show that $V_{2k+1}$ is $G$-equivalent to the identity.
This follows from the facts that the $\ZZ$-homology chain $C_{2k+1}$ 
is integrally closed (\cref{lem:HalperinToledo}) and  $\omega_{2k+1}$ is trivial in $H^{2k+1}(\ZZ_2;\RR/\ZZ)$.
\cref{lem:coboudarylocalsymmetric} concludes the argument.
The second claim follows from \cref{lem:invertible} since $C_k$ are closed modulo $2$.

If $d$ is odd,
we express the cochain circuit $V_{2k+1}$ using $\frac 1 2 \delta \omega_{2k}$;
as in the proof of \cref{lem:coboudarylocalsymmetric}, this gives a locally $G$-symmetric circuit
and a cochain circuit by $\frac 1 2 \omega_{2k}$ on $\partial C_{2k+1}$.
\cref{lem:HalperinToledo} says $\partial C_{2k+1} = 2 C_{2k}$.
Hence, modulo a locally $G$-symmetric circuit, $V_{2k+1}$ is equal to $V_{2k}$.
Therefore $V_{2k+1} V_{2k}$ is $G$-equivalent to the identity.
\end{proof}

When the combinatorial manifold $L$ has non-zero Stiefel-Whitney classes,
the result of \cref{thm:GDSwoTR} suggests that the GDS dual on $L$ 
is not constant depth circuit-equivalent to just a cohomology model in the top dimension.
We do not prove this rigorously in general, 
but let us discuss the special case $L={\mathbb C}P^2$ ($d=4$).
First, \cref{thm:GDSwoTR} immediately shows that 
the GDS dual ground state on ${\mathbb C}P^2$ has odd $\ZZ_2$ charge~\cite{GDS, Debray2018},
so certainly there can be no symmetric circuit that turns it into the $\ZZ_2$-even ground state of a cohomology model.
However, let us avoid this trivial obstruction by removing the $\ZZ_2$ charge from GDS.
The resulting state is then a 4-dimensional cohomology model stacked with a 2-dimensional cohomology model 
on $S^2 = \mathbb C P^1 \subset \mathbb C P^2$.
Here $S^2$, the $2$-cell in the usual CW complex construction of ${\mathbb C}P^2$, 
is a generator of the second homology group of $\mathbb CP^2$
and is in fact a representative of the Stiefel-Whitney homology class of ${\mathbb C}P^2$.
Let us argue that this state is not equivalent to just a 4-dimensional cohomology model.
Equivalently, we will show:

\begin{claim} \label{physargument2}
Under the physical assumption that there are no non-trivial invertible states 
in $3$-spatial dimensions and no non-trivial $3$-dimensional SPTs 
of unitary onsite $\ZZ_2$ symmetry, 
the $2d$ cohomology state on $S^2 \subset {\mathbb C}P^2$ cannot be disentangled with a symmetric circuit.
\end{claim}
In certain settings, results of this kind are easy to prove.
For example, suppose we instead had $T^4 = T^2 \times T^2$, 
and put a 2d cohomology phase on one of the $T^2$'s.
Then any putative symmetric local circuit that disentangles this state 
must still be local after dimensionally reducing along the other $T^2$ 
(i.e., tensoring all the site Hilbert spaces on each $T^2$ fiber into a single super-site).
This dimensional reduced circuit disentangles a 2d cohomology SPT 
and, after gauging, we get a circuit that maps the double semion to the toric code, which is not possible~\cite{Haah2014invariant}.

In the present case of $S^2 \subset {\mathbb C}P^2$, 
this kind of dimensional reduction is not possible.
Instead, our argument will rely on the following fact: 
puncturing ${\mathbb C}P^2$ by removing a point 
results in the total space of a complex line bundle over $\mathbb C P^1$ with a non-trivial Chern number.
($\mathbb C P^2 \setminus \{pt\}  = \{ [z_0;z_1;z_2] \} \setminus \{ [1;0;0] \} \to \{[z_1;z_2]\}$
is the bundle projection map.)
By removing a small ball around the point,
we can view ${\mathbb C}P^2$ as being made from gluing two components 
$X$ and $Y$ along their common boundary $S^3$
\begin{align}
{\mathbb C}P^2 = X \coprod_{S^3} Y 
\end{align}
where $X$ is a $D^2$ disc bundle over $S^2$ and $Y$ is a $4$-ball.
It will be useful for us to think of the $S^3$ as being slightly thickened, 
i.e., having a collar, whose size is small compared to the size of ${\mathbb C}P^2$ 
but large compared to all microscopic length scales, 
including the range of all constant depth circuits we consider.
Let ${\tilde{X}}$ and ${\tilde{Y}}$ be $X$ and $Y$ thickened by this collar.

Now assume, for a contradiction, that a constant depth circuit $U$ on ${\mathbb C}P^2$, 
made out of $\ZZ_2$-symmetric gates, acts on a trivial product state $\ket +$ 
to create a non-trivial cohomology state on $S^2 \subset X \subset {\mathbb C}P^2$ 
tensored with a trivial product state away from the $S^2$.
We can break up $U = U_{\tilde{X}} U_{\tilde{Y}}$, 
where $U_{\tilde{X}}$ and $U_{\tilde{Y}}$ are symmetric circuits 
acting on ${\tilde{X}}$ and ${\tilde{Y}}$, respectively.
Consider now $\ket \psi = U_{\tilde{Y}}\ket +$,
which must be equal to some state $\ket{\psi}_{\tilde{X}}$ on ${\tilde{X}}$,
tensored with a trivial product state on ${\mathbb{C}}P^2 \setminus {\tilde{X}}$, 
since this is true for $U_{\tilde{X}} \ket \psi = U \ket +$ by assumption, 
and $U_{\tilde{X}}$ acts as the identity on ${\mathbb{C}}P^2 \setminus {\tilde{X}}$.  
But $U_{\tilde X}$ is supported on the disc bundle over $S^2$,
which deformation-retracts onto $S^2$.
So $U_{\tilde X}$ maps $\ket{\psi}_{\tilde X}\otimes |+\rangle_{{\mathbb C}P^2\setminus \tilde X}$ to $U \ket +$,
where the latter is the nontrivial 2d cohomology state on $S^2$.

However, we know that $\ket \psi_{\tilde{X}}$ looks like $\ket \psi = U_{\tilde{Y}} \ket +$, 
i.e., a trivial symmetric product state, on ${\mathbb{C}}P^2 \setminus {\tilde{Y}}$.
Thus $\ket \psi_{\tilde{X}}$ is the tensor product of some state $\ket \psi_{S^3}$ defined in the collar of the $S^3$,
with a trivial product state.  Now, $\ket \psi$ has a symmetric gapped parent Hamiltonian on ${\mathbb C}P^2$ obtained from conjugating the trivial sum of projectors Hamiltonian by $U_{\tilde{Y}}$.  Since this parent Hamiltonian stabilizes a trivial symmetric product state away from the collar, it can be deformed, in the space of symmetric gapped Hamiltonians, to a trivial sum of projectors there.  The remaining terms in the collar give a quasi $3d$ symmetric gapped parent Hamiltonian for $\ket \psi_{S^3}$.  Because $\psi_{S^3}$ was obtained by applying a constant depth circuit in ${\tilde{Y}}$ (``pumping a state out to the boundary of ${\tilde{Y}}$"), it must be an invertible state.

We now use some beliefs about the classification of short range entangled states.
First, it is believed that there are no non-trivial invertible states in $3$ spatial 
dimensions~\cite{Kitaev_pc, Freed, FreedHopkins, GaiottoFreyd}, 
so $\ket{\psi}_{S^3}$ must actually be a $3d$ symmetry protected state.
The classification of $\ZZ_2$ symmetry protected states in $3d$ is trivial~\cite{Wen, Kapustin_cobordism},
so we can apply a symmetric circuit to turn $\ket{\psi}_{S^3}$ 
into a product state in the interiors of the $3$-cells of a triangulation of $S^3$,
resulting in another $\ZZ_2$-symmetric state $\ket {\psi'}_{S^3}$.
But $\ket{\psi'}_{S^3}$ is supported on some subdimensional set of $S^3$
that does not cover all of $S^3$ and hence is null-homologous.
This means that $\ket{\psi'}_{S^3}$ 
can in fact be trivialized by locally nucleating small bubbles of $\ZZ_2$ SPTs.
This implies that the trivial state $\ket{\psi}_{S^3}$ is transformed into the nontrivial 2d cohomology state
via $U_{\tilde X}$, which is absurd.

\begin{acknowledgments}
We thank Michael Freedman for his explanation of Stiefel-Whitney classes and many other useful discussions. 
We thank Anton Kapustin, Dan Freed, and Arun Debray for useful discussions.
N.T. acknowledges the support of the Natural Sciences and Engineering Research Council of Canada (NSERC).
L.F. acknowledges NSF DMR 1519579.
\end{acknowledgments}

\bibliography{gds-ref}
\end{document}